%% file: final-v2.tex
\documentclass[pra, twocolumn, showkeys, nofootinbib, twoside]{revtex4-2}

\input{packages.tex}

\usepackage{breqn}
\usepackage{subcaption}

\usepackage{bm}

\newcommand{\UmUU}{U(-\btheta)\otimes U(\btheta)^{\otimes 2}}
\newcommand{\UmUUdag}{U(\btheta)\otimes U(-\btheta)^{\otimes 2}}

\newcommand{\btheta}{\text{\boldmath$\theta$}}

\newcommand{\bx}{\mathbf{x}}

\newcommand{\msf}{\mathsf}

\newcommand{\pure}{\text{pure}}
\renewcommand{\phase}{\btheta}
\newcommand{\ideal}{\text{ideal}}
\newcommand{\proc}{\text{proc}}
\newcommand{\cptp}{\text{CPTP}}

\definecolor{cool_green}{rgb}{0.0, 0.5, 0.0}

\begin{document}

\title{Process-optimized phase covariant quantum cloning}
\author{Chloe Kim}
\affiliation{University of Illinois at Urbana-Champaign}
\author{Eric Chitambar}
\affiliation{University of Illinois at Urbana-Champaign}
\date{\today}

\begin{abstract}
After the appearance of the no-cloning theorem, approximate quantum cloning machines (QCMs) has become a well-studied subject in quantum information theory.
Among several measures to quantify the performance of a QCM, single-qudit
fidelity and global fidelity have been most widely used. In this paper 
we compute the optimal global fidelity for phase-covariant cloning machines via semi-definite programming optimization, thereby completing a remaining gap in the previous results on QCMs.  We also consider optimal simulations of the cloning and transpose cloning map, both by a direct optimization and by a composition of component-wise optimal QCMs.  For the cloning map the composition method is sub-optimal whereas for the transpose cloning map the method is asymptotically optimal.
\end{abstract}

\maketitle

\section{Introduction}
\tab The famous no-cloning theorem tells us that
there exists no quantum operation that copies an arbitrary quantum state
\cite{Park1970,Dieks1962,Wootters1982}.  Despite this impossibility, the notion of approximate quantum cloning machines was introduced in 1996 \cite{Buzek1996}, and the topic has received extensive development over the past twenty-five years.  The inability to universally copy states is a fundamental signature of ``quantumness" that has significant applications.  In particular, the no-cloning theorem is essential for ensuring security in cryptographic tasks like quantum key distribution (QKD) \cite{Bennett1984}.  A general overview of quantum cloning and its relevance to quantum information processing can be found in review papers such as \cite{Scarani2005, Fan2014}.

An ideal $N\to M$ quantum cloner is a non-linear map
\begin{equation}
  \rho^{\otimes N} \mapsto \rho^{\otimes M}
\end{equation}
for any density matrix $\rho$ of some $d$-dimensional system $\msf{S}$, and a quantum cloning machine (QCM) approximates this operation via a physically-realizable quantum
channel.  That is, an $N\to M$ QCM is a completely positive trace-preserving (CPTP) map $\mc{E}$ satisfying
\begin{align}
\mc{E}\left(\rho^{\otimes N}\right)=\sigma^{\msf{S}_1\cdots\msf{S}_M}\approx \rho^{\otimes M}.
\end{align}
Different criteria are used to categorize QCMs.  For instance, the QCM is called symmetric if $\sigma^{\msf{S}_i}:=\tr_{\overline{\msf{S}}_i}\sigma$ is the same reduced state for all $\msf{S}_i$, where $\tr_{\overline{\msf{S}}_i}$ indicates a partial trace over all subsystems but $\msf{S}_i$.  An economic QCM is one that restricts to maps using no ancillary system other than the $\msf{S}_i$
output systems \cite{Buscemi2005}.  Finally, universality of a QCM means that it clones all the input states equally well, and a QCM is optimal
if the cloning process is simulated optimally with respect to a given figure of merit \cite{Scarani2005}.  

The most commonly used figure of merit for cloning is the average single-qudit fidelity, which measures the distance between the input state $\rho$ and the reduced state of each output
copy $\sigma^{\msf{S}_i}$.  Another possible measure evaluates the average global fidelity between the QCM output state $\sigma^{\msf{S}_1\cdots\msf{S}_M}$ and the ideal output state
$\rho^{\otimes M}$. In \cite{Koniorczyk2013a}, this global fidelity was called
\emph{process fidelity}, and the authors used the semidefinite programming (SDP)
techniques introduced in \cite{Audenaert2002a} to optimize the process fidelity
for cloning one qubit ($d=2$).  The goal of this paper is to extend this result by computing the optimal process fidelity for cloning qudits ($d>2$).  In particular, we consider the cloning of maximally coherent pure states having the form
\begin{equation}\label{phase covariant-state}
\ket{\btheta}=\frac{1}{\sqrt{d}}\sum_{k=1}^de^{i\theta_k}\ket{k},
\end{equation}
where $\btheta$ denotes $d$-tuple of angles $(\theta_1,\cdots,\theta_d)$, $\theta_i\in \left[0,2\pi\right)$.  For $d=2$, this represents the family of equatorial states on the Bloch sphere.  Up to a unitary transformation, the four BB84 states $\{\ket{0},\ket{1},\ket{\pm}=\frac{1}{\sqrt{2}}(\ket{0}\pm\ket{1})\}$ belong to this family \cite{Bennett1984}.
\tab 

A (non-universal) QCM whose domain is restricted to the phase-encoded input states $\ket{\btheta}$  is called a phase-covariant quantum cloner.
Bru\ss{} et al. first studied the cloning of BB84 states \cite{Bruss2000a},
and they showed that the optimal map that clones BB84 states also optimally
clones any arbitrary phase-covariant states on the Bloch sphere.
They obtained an optimal single-qubit fidelity of
$\frac{1}{2} + \frac{1}{\sqrt{8}}$, which is greater than that of the universal
cloner ($\frac{5}{6}$).
Subsequent work has been conducted on the phase-covariant cloning of
$1\to M$ qubits \cite{Fan2001b},
$1\to 2$ qutrits \cite{DAriano2003, DAriano2001,Cerf2002b},
$1\to M$ qudits \cite{FanEtAl2003, Buscemi2007}, 
$N\to M$ qubits \cite{Fiurasek2004},
$N\to M$ qudits for $M = kd + N$ \cite{Buscemi2007},
optimizing over the single qudit fidelity (See \cite{Buscemi-PhD} for more).

In this paper, we present an explicit formula for an optimal $1\to 2$ phase covariant QCM,
using the process
fidelity as the figure of measure.  This completes the picture of $1\to 2$ optimal QCMs, as summarized in Table \ref{tab:QCM-results}.
The organization of this paper as follows.  In Section \ref{Sect:notation} we specify
notation and definitions related to the process fidelity.  In Section \ref{Sect:phase-covariant} we present
our main results on optimal process fidelity of phase-covariant cloning. We introduce
transpose cloning in Section \ref{Sect:transpose-cloning} and provide the optimal single-qudit and process fidelities.  This allows us to consider in Section \ref{Sect:modular} the  modular construction of larger QCMs by composing smaller ones.  Finally, some concluding remarks are provided in Section \ref{Sect:conclusion}.

\onecolumngrid
\vspace{1em}


{\renewcommand{\arraystretch}{1.15}
\begin{table}[h!]
\centering
\begin{tabular}{|c|c|c|c|}
\hline
\thead{dimension} & \thead{type} & \thead{optimal single-qubit fidelity} & \thead{optimal process fidelity}\\
\hline
\hline
$1\to 2$ qubits & universal & $\frac{5}{6}$ \cite{Buzek1996, Bruss1998a, Gisin1997, Gisin1999a} & $\frac{2}{3}$\cite{Buzek1996}\\
\hline
$1\to 2$ qudits & universal & $\frac{d+3}{(2d+2)}$ \cite{Buzek1999} & $2/(d+1)$\cite{Werner1998}\\
\hline
$1\to 2$ qubits & phase covariant & $\frac{1}{2} + \frac{1}{\sqrt{8}}$ \cite{Acin2004b, Bruss2000a, Durt2004, Griffiths1997} & 0.75 \cite{Koniorczyk2013a}\\
\hline
$1\to 2$ qutrits & phase covariant & $\frac{5+\sqrt{17}}{12}$\cite{DAriano2003, DAriano2001,Cerf2002b} & $5/9\;$ (this paper)\\ 
\hline
$1\to 2$ qudits & phase covariant & $\frac{1}{d} + \frac{d-2+\sqrt{d^2+4d-4}}{4d}$ \cite{FanEtAl2003} & $(2d-1)/d^2\;$ (this paper)\\ 
\hline
\end{tabular}
\caption{Optimal single qudit fidelity and process fidelity of symmetric UQCM and phase covariant QCM}
\label{tab:QCM-results}
\end{table}
}
\vspace{20pt}

\twocolumngrid

\section{Notations and definitions}

\label{Sect:notation}
Let $B(\cH)$ denote the set of bounded operators on a Hilbert space $\cH$.
Let $\cD \subset B(\cH)$ be the set of density operators, i.e. positive trace-one
elements, $\cD_{\pure} \subset \cD$ the set of pure states, i.e. rank-one density operators,
and $\cD_{\phase} \subset \cD_{\pure}$ the set of states $\op{\btheta}{\btheta}$ having the form of \eqref{phase covariant-state}.
The fidelity \cite{Jozsa-1994a} of any two $\rho,\sigma \in \cD$ is given by
\begin{equation}
  F(\rho, \sigma) := \( \tr \sqrt{\sqrt{\rho}\sigma\sqrt{\rho}} \)^2.
\end{equation}

Suppose $\cE_{\ideal}: B(\cH)\to B(\cK)$ is an ``ideal'' map that is not
necessarily a quantum channel (not even necessarily linear), i.e,
completely-positive and trace-preserving (CPTP).  We are interested in
approximating $\cE_{\ideal}$ by a physically-realizable quantum channel $\mc{E}$,
or possibly approximating just the action of $\mc{E}_{\ideal}$ on some restricted
set of inputs $S$. There are various approaches to quantifying how well $\mc{E}$
approximates $\cE_{\ideal}$, and here we consider the process fidelity.


\begin{definition}
Let $S\subset \cD$ be the set of states that are generated by the action of a compact group $G$,
and let $\mu$ be the Haar measure induced by $G$.
The \ti{process fidelity} between an arbitrary ideal map $\cE_{\ideal}$ and a CPTP map $\cE$ is defined as
\begin{equation}
F_{\proc}(\cE_{\ideal}, \cE|S) := \int_{S} F(\cE_{\ideal}(\sigma), \cE(\sigma)) d\mu(\sigma).
\end{equation}
\end{definition}
If $\cE_{\ideal}$ maps pure state to pure states and $S\subset\mc{D}_\pure$,
then the process fidelity is equivalent to
\noeqref{Eq:proc-pure}
\begin{equation}\label{Eq:proc-pure}
F_{\proc}(\cE_{\ideal}, \cE| S) = \int_{S} \tr \Big[ \cE_{\ideal}(\sigma) \cE(\sigma) \Big] \;d\mu(\sigma).
\end{equation}
A key feature of this expression is that the integrand is linear in $\mc{E}(\sigma)$.  This will allow us to exploit symmetry properties of the Haar measure below.  On the other hand, note that $\cE_{\ideal}$ need not be linear, such as with the cloning map $\rho\mapsto \rho\otimes \rho$.

\begin{definition}
Let $\cptp(\mc{H}\to\mc{K})$ be the set of all quantum channels mapping $B(\cH)$ to $B(\cK)$.  The \ti{optimal process fidelity} of a map $\mc{E}_\ideal:B(\mc{H})\to B(\mc{K})$ is the maximum process fidelity
over all possible quantum channels, i.e.
\begin{equation}
F^*_{\proc}(\cE_{\ideal}|S) := \max_{\cE\in \cptp(\mc{H}\to\mc{K})} \; F_{\proc}( \cE_{\ideal}, \cE|S).
\end{equation}
\end{definition}
Given an orthonormal basis $\{\ket{i}\}_{i=1}^d$ for $\mc{H}$, the Choi-Jamio\l{}kowski isomorphism \cite{Jamiolkowski-1972a, Choi-1975a} establishes an equivalence between every channel $\cE\in \cptp(\mc{H}\to\mc{K})$ and an operator $J(\mc{E})\in B(\mc{H})\otimes B(\mc{K})$,
\noeqref{choi-iso}
\begin{equation}\label{choi-iso}
\cE \leftrightarrow J(\mathcal{E}):=\sum_{i,j} \op{i}{j}\otimes \cE(\op{i}{j}).
\end{equation}
The operator $J(\cE)$ is called
the Choi matrix of $\cE$ and the action of $\cE$ can be directly expressed
in terms of $J(\mc{E})$ as
\begin{equation}
\cE(\rho) = \tr_{\mc{H}} (J(\cE)(\rho^{T}\otimes \1)),
\end{equation}
where $\1$ is the identity on $\mc{K}$ and $\rho^T$ is the transpose of $\rho$ in some fixed basis.  When $S\subset \cD_{\pure}$ and $\cE_{\ideal}: \cD_{\pure}\to \cD_{\pure}$, the process fidelity can be expressed in terms of the Choi matrix as
\begin{align}
  F_{\proc}(\cE_{\ideal}, \cE| S)= \int_S \tr\Big[ J(\cE) \rho^{T}\otimes \cE_{\ideal}(\rho) \Big] \; d\mu(\rho).\qquad\label{proc-fid-choi}
\end{align}
To compute $F_{\proc}^*(\cE_{\ideal}, \cE|S)$, the problem then reduces to finding  the Choi matrix $J(\cE)$ that maximizes Eq. \eqref{proc-fid-choi}.

\section{Process-optimized phase-covariant quantum cloning}

\label{Sect:phase-covariant}

\subsubsection{Characterization of an average Choi map}
The $1\to 2$ ideal phase covariant cloner $\cE_{\ideal}$ is given by
\eqref{Eq:ideal-map}
\begin{equation}\label{Eq:ideal-map}
  \cE_{\ideal}(\dop{\btheta}) = \dop{\btheta}^{\otimes 2}
\end{equation}
for an arbitrary $\dop{\btheta}\in \cD_{\phase}$.  The goal is to find an optimal quantum channel $\cE$ satisfying
\begin{align}
  F_{\proc}(\mc{E}_\ideal,\mc{E}|\mc{D}_{\phase}) = F^*_{\proc}(\cE_{\ideal}|\mc{D}_{\phase}).
\end{align}
It is useful to write $\ket{\btheta}$ as $U(\btheta)\ket{\phi^+_d}$,
where $U(\btheta)=\sum_{k=1}^de^{i\theta_k}\op{k}{k}$ and 
$\ket{\phi^+_d}=\frac{1}{\sqrt{d}}\sum_{k=1}^d\ket{k}$.  
By the definition expressed in \eqref{proc-fid-choi}, the process fidelity for
$\cE$ can be written as
\noeqref{eq:proc-def-choi}
\begin{equation}\label{eq:proc-def-choi}
  F_{\proc}(\cE_{\ideal}, \cE| \mc{D}_{\phase}) =\tr\left[\dop{\phi_d^+}^{\otimes 3}\mc{T}(J(\cE))\right],
\end{equation}
where $\cT(\cdot)$ denotes the ``twirling'' operation
\begin{small}
\begin{equation}
  \label{phase-anti-proc}
  \mc{T}(X):=\int \UmUUdag\left(X\right)\UmUU \;d\mu(\btheta).
\end{equation}
\end{small}
We exploit the symmetric properties of $\cT(J(\cE))$ to formulate the constraints
of a semidefinite program to obtain maximum process fidelity. This is a standard
trick \cite{Vollbrecht-2001a}, but we include it as lemmas for the sake of self-containment.

Let $S_d$ be the permutation group of $d$ elements.
Given an orthonormal basis $\{\ket{k}\}_k$, define
\begin{equation}
  U_\pi = \sum_{k=1}^d \op{\pi(k)}{k}, \qquad\pi\in S_d.
\end{equation}
In other words, $U_\pi$ permutes the basis vectors given a permutation
$\pi \in S_d$.
Next, let $V_\sigma$ be the operator permuting three sub-systems, where $\sigma\in S_3$.
For example,
\begin{equation}
  V_{(23)}\ket{\phi_1}\otimes\ket{\phi_2}\otimes \ket{\phi_3} = \ket{\phi_{1}} \otimes \ket{\phi_{3}} \otimes \ket{\phi_{2}}.
\end{equation}

\begin{lemma}\label{Lem:avg-choi}
  Suppose $\cE$ is a quantum channel with Choi matrix $J(\cE)$.
  If we define the average Choi map for the channel $\cE$ as
  \begin{equation}
    \label{avg-choi-equatorial}
    \wt{J}(\cE) := \frac{1}{2\abs{S_d}}\sum_{\substack{\pi\in S_d\\\sigma\in\{\id,(23)\}}}      V_\sigma U_\pi^{\otimes 3} \cT(J(\cE)) U_\pi^{\dagger \otimes 3} V_\sigma^{\dagger},
  \end{equation}
  then we have
  \begin{enumerate}[(i)]
    \item $\wt{J}(\cE)$ is invariant under conjugation by $\UmUU$
      for any $\btheta$;\label{Eq:lem-Utheta}
    \item $\wt{J}(\cE)$ is invariant under conjugation by $U_\pi^{\otimes 3}$
    for all $\pi \in S_d$;\label{Eq:lem-Upi}
    \item $\wt{J}(\cE)$ is invariant under conjugation by $V_{(23)}$;\label{Eq:lem-V23}
    \item $\wt{J}(\cE)$ is positive and $\tr_{23}(\wt{J}(\cE)) = \1$.\label{Eq:lem-pos}
  \end{enumerate}
\end{lemma}

Note that property (iv) of Lemma \ref{Lem:avg-choi} assures that $\wt{J}(\mc{E})$ is a valid Choi matrix of a quantum channel.  
\begin{lemma} \label{Lem:proc}
  Let $\cE_{\ideal}$ be the ideal
  phase-covariant cloning map and $\mc{E}$ an arbitrary quantum channel.  Then
  \noeqref{Eq:proc-avg-choi}
  \begin{equation}\label{Eq:proc-avg-choi}
    F_{\proc}(\cE_{\ideal}, \cE | \mc{D}_{\phase}) = \tr\left[\phi_d^{+\otimes 3} \wt{J}(\mc{E})\right].
  \end{equation}
 \end{lemma}
  \begin{proof}
    Let $\pi' \in S_d$. We can observe that
    \begin{equation}
      \cT(X) = U_{\pi'}^{\dagger\otimes 3} \cT(U_{\pi'}^{\otimes 3}X U_{\pi'}^{\dagger\otimes 3}) U_{\pi'}^{\otimes 3}.
    \end{equation}
    Using (ii) of the Lemma \ref{Lem:avg-choi}, it follows that
    \begin{small}
    \begin{align}
      &F_{\proc}(\cE_{\ideal}, \mc{E}|\mc{D}_{\phase})= \tr\left[\phi_d^{+\otimes 3} \cT(J(\cE))\right]\\
        &=\tr\left[\phi_d^{+\otimes 3} U_{\pi'}^{\dagger\otimes 3} \cT(U_{\pi'}^{\otimes 3}J(\cE) U_{\pi'}^{\dagger\otimes 3}) U_{\pi'}^{\otimes 3}\right]\notag\\
        &= \tr\left[\phi_d^{+\otimes 3}\frac{1}{\abs{S_d}}\sum_{\pi\in S_d}U_\pi^{\otimes 3}\cT(J(\cE))U_\pi^{\dagger\otimes 3}\right]\notag\\
        &=\tr\left[\phi_d^{+\otimes 3} \frac{1}{2\abs{S_d}}\sum_{\substack{\pi\in S_d\\\sigma\in\{\id,(23)\}}} V_\sigma U_\pi^{\otimes 3} \cT(J(\cE)) U_\pi^{\dagger \otimes 3} V_\sigma^{\dagger}\right]\\
        &=\tr\left[\phi_d^{+\otimes 3} \wt{J}(\mc{E})\right],
    \end{align}
  \end{small}
    where we have used the facts that $\phi_d^+$ is invariant under $\pi'$ and  $\phi_d^{+\otimes 3}$ is invariant under $V_{(23)}$.
  \end{proof}
In summary, we have
\begin{align}
\label{Eq:proc-fidelity2}
F^*_{\proc}(\mc{E}_{\ideal}|\mc{D}_\phase)&=\max \tr\left[\phi_d^{+\otimes 3}X\right]
\end{align}
in which the maximization is taken over all operators $X\in B(\mbb{C}^d)^{\otimes 3}$ satisfying properties \ref{Eq:lem-Utheta}, \ref{Eq:lem-Upi}, \ref{Eq:lem-V23}, \ref{Eq:lem-pos} in 
Lemma \ref{Lem:avg-choi}.

\subsubsection{Optimization via semidefinite programming}
\tab Using the invariant properties of the average Choi matrix for quantum
channel $\cE$ in Lemma \ref{Lem:proc}, we will construct a semi-definite program (SDP) to obtain an optimal $\mc{E}$
that maximizes the process fidelity.  A particularly nice reference for applying semidefinite programming to quantum information problems is \cite{Watrous-2018a}, and we apply the basic results here.

First, we characterize a $d^3\times d^3$ hermitian operator $X\neq 0$ that satisfies 
\ref{Eq:lem-Utheta}, \ref{Eq:lem-Upi}, \ref{Eq:lem-V23}, \ref{Eq:lem-pos} in 
Lemma \ref{Lem:avg-choi}. If we write
\begin{equation}
  X=\sum_{i,j,k,l,m,n}x_{ijklmn}\op{ijk}{lmn},
\end{equation}
then the invariance under $\UmUU$ gives
\begin{equation}
  \label{Eq:phase-invariance}
  x_{ijklmn}(1-e^{i(-\theta_i+\theta_j+\theta_k+\theta_l-\theta_m-\theta_n)})=0.
\end{equation}
This can be satisfied if and only if $x_{ijklmn}=0$ or 
$-\theta_i+\theta_j+\theta_k+\theta_l-\theta_m-\theta_n$ is identically zero. If we next add the invariance under under $V_{(23)}$ and $U_\pi^{\otimes 3}$, then $X$ has the form $\sum_{i=1}^9 x_iX_i$, where $x_i \in \C$ and
\begin{small}
\begin{align}
  &X_1 = \sum_{i} \op{iii}{iii},\;
  X_2 = \sum_{i\neq k} \op{iik}{iik} + \op{iki}{iki},\\
  &X_3 = \sum_{i\neq k} \op{kii}{kii},\\
  &X_4 = \sum_{i\neq k} \op{kik}{iii} + \op{iii}{kik} + \op{kki}{iii} + \op{iii}{kki},\\
  &X_5 = \sum_{i\neq k}  \op{iik}{iki} + \op{iki}{iik},\;
  X_6 = \sum_{i\neq k\neq \ell} \op{ik\ell}{ik\ell},\\
  &X_7 = \sum_{i\neq k\neq \ell} \op{kk\ell}{i\ell i} + \op{\ell k\ell}{iik},\\
  &X_8 = \sum_{i\neq k\neq \ell} \op{kk\ell}{ii\ell} + \op{\ell k\ell}{iki},\;
  X_9 = \sum_{i\neq k\neq \ell} \op{ik\ell}{i\ell k}.
\end{align}
\end{small}
If we denote $\mathbf{x} = (x_i) \in \C^9$, then 
\small
\begin{align}
  \tr\left[\phi_d^{+\otimes 3} X\right]= \frac{1}{d^2}x_1 + \frac{(d-1)}{d^2} (2x_2 + x_3 + 4x_4 + 2x_5)\\
  +\frac{(d-1)(d-2)}{d^2}(x_6 + 2x_7 + 2x_8 + x_9).
\end{align}
The condition $\tr_{23} X=\1$ corresponds to the equality
\small
\begin{equation}
  A(\mbf{x}) := x_1 + (d-1)(2x_2 + x_3) + (d-1)(d-2)x_6 = 1.
\end{equation}
Finally, the positivity constraint $X \geq 0$ allows us to express \eqref{Eq:proc-fidelity2} in a simplified SDP.  Let $\mathbf{a} = (a_i)$ where $a_i$ is the coefficient of $x_i$'s in $A(\mathbf{x})$,
and $\mathbf{c} = (c_i)$ be the, the coefficient of $x_i$ in $F_\cE(\mathbf{x})$.
Define $F_0 = 0_{d^3\times d^3} \oplus [1]$, $F_i = X_i \oplus [-a_i].$
Then we have the primal form of the SDP
\begin{align} \label{SDP}
  &\text{minimize} \q -\mathbf{c}^T\bx \\
  &\text{subject to} \q F_0 + \sum_i x_iF_i \geq 0,\q \mathbf{a}^T\bx = 1.
\end{align}
By Eq. \eqref{Eq:proc-fidelity2}, this SDP yields the value of $F^*_\proc(\mc{E}_\ideal|\mc{D}_\phase)$.

Now we present our main results.
\begin{theorem} 
\label{Thm:1-2-cloner}
Let $\cE_{\ideal}$ be the ideal $1\to 2$ phase covariant cloner.
  Then, the optimal process fidelity of $\cE_{\ideal}$ is
  \begin{equation}
    F^*_{\proc}(\cE_{\ideal}|\cD_{\phase}) = \frac{2d-1}{d^2},
  \end{equation}
  where $d$ is the dimension of the input system.
\begin{proof}
  Define \begin{gather}
    k_d = \frac{1}{2d-1},\\
    \bx = (k_d, k_d, 0, k_d, k_d, 0, k_d, k_d, 0),
  \end{gather}
  where $d$ is the dimension. 
It is straightforward to show that $X = \sum_i x_iX_i$ satisfies the positivity and the trace
  condition. So this $\bx$ is a primal feasible solution and it yields
  $F_\cE(\bx) = \frac{2d-1}{d^2}$.
  Consider the dual form
  \begin{align} \label{SDP}
    &\text{maximize} \q -\tr F_0Z\\
    &\text{subject to} \q\tr F_iZ = -c_i, Z \geq 0.
  \end{align}
  We will construct $Z$
  such that $\tr F_0Z = (2d-1)/d^2$ and show that this $Z$ is dual feasible.  By strong duality, this implies that $(2d-1)/d^2$ is indeed the optimal solution.
  Let $Z = \hat{Z}\oplus z$, where
  $\hat{Z} = \sum_i b_iX_i$, $b_i\in \R$ and $z = \frac{2d-1}{d^2}$. Then
  \begin{equation}
    \tr F_0Z = \frac{2d-1}{d^2},
  \end{equation}
  and the constraints of the dual form become
  \begin{align}
    \begin{aligned}
      &\tr[X_1\hat{Z}] = z - \frac{1}{d^2},\\
      &\tr[X_2\hat{Z}] = 2\tr[X_3\hat{Z}] = 2(d-1)z -\frac{2(d-1)}{d^2}\\
      &\tr[X_4\hat{Z}] = 2\tr[X_5\hat{Z}] = -\frac{4(d-1)}{d^2},\\
      &\tr[X_6\hat{Z}] = (d-1)(d-2)z -\frac{(d-1)(d-2)}{d^2}\\
      &\tr[X_7\hat{Z}] = \tr[X_8\hat{Z}] = 2\tr[X_9\hat{Z}] = -\frac{2(d-1)(d-2)}{d^2}.
    \end{aligned}
  \end{align}
  Solving these gives $b_1 = b_2 = b_3 = b_6 = \frac{2(d-1)}{d^3},$ $b_4 = b_5 = b_7 = b_8 = b_9 = -\frac{1}{d^3}.$
  We can rewrite $\hat{Z}$ as a linear combination of projections 
\begin{align}
      \hat{Z} &= \frac{1}{d^3}\Big[
        (2d-2-\sqrt{2})X_1 + \left(2d-\frac{9}{2}-\frac{3}{2\sqrt{2}}X_2 \right) \\
        &+ (2d-2)X_3 + (2d-3)X_6 +2\sqrt{2}P_A \\
        &+ 2\left(\frac{2}{\sqrt{2}}+1\right)P_B + 2P_C + 2P_D + 2P_E \Big],
  \end{align}
  where we define
  \begin{align}
      P_A &= -\frac{1}{2\sqrt{2}}X_4 + \frac{1}{4}(X_1+X_2) + \frac{1}{4}(X_1+X_5)\notag\\
      P_B &= \frac{1}{2}(X_2-X_5),\quad
      P_C = \frac{1}{2}(X_2-X_7),\quad\\
      P_D &= \frac{1}{2}(X_2-X_8),\quad
      P_E = \frac{1}{2}(X_6-X_9).
  \end{align}
  When $d\geq 3$ all the coefficients are positive, and when
  $d = 2$ it can be easily checked that $\hat{Z}\geq 0$. This shows that $Z$ is also
  positive, and therefore $\frac{2d-1}{d^2}$ is indeed the optimal solution.
\end{proof}
\end{theorem}

Note that one can calculate the single-qudit fidelity using a map $\cE$ that is optimal for the process fidelity.  Doing so leads to a single-qudit fidelity of
$(d+1)/(2d-1)$, which is strictly smaller than the optimal value \cite{FanEtAl2003}.
This means that the map yielding optimal process fidelity does not
give optimal single-qudit fidelity unlike the universal cloners \cite{Werner1998}.

\section{Transposition cloning}

\label{Sect:transpose-cloning}
\subsection{Optimal phase-covariant transposition}
\tab The transposition map is another famous non-CP map. The ideal operation
of transposition map is given by
\begin{equation}
\label{Eq:transpose}
  \mc{E}_\ideal:\op{\psi}{\psi} \mapsto \op{\psi}{\psi}^T,\qquad\forall\rho\in\mc{D}_\pure,
\end{equation}
It has been shown that the optimal fidelity for approximate transposition
channel of arbitrary pure states is $\frac{2}{d+1}$ \cite{Buscemi2003}. What happens if we
restrict the input states to phase covariant states?  In this case, Eq. \eqref{Eq:transpose} takes the form
\begin{equation}
\label{Eq:phase-covariant-transpose}
  \mc{E}_\ideal:\op{\btheta}{\btheta} \mapsto \op{\btheta}{\btheta}^T=\op{-\btheta}{-\btheta},\qquad\forall\rho\in\mc{D}_\phase,
\end{equation}
  where $-\btheta = (-\theta_1, \cdots, $ $-\theta_d)$. To compute the optimal process fidelity in this restricted case, we
can use the same SDP analysis. The calculation yields to $F_{\proc} = \frac{2}{d}$, and the Choi operator of an optimal map is given as
\begin{equation}
  \wt{J}(\cE) = \frac{1}{d-1} \sum_{i\neq j}(\op{ij}{ij} + \op{ij}{ji}).
\end{equation}
Hence we see that restricting to maximally coherent states $\ket{\btheta}$ only yields an improved process fidelity of $\frac{2}{d}$ versus $\frac{2}{d+1}$ in the universal case.  These have the same asymptotic scaling, yet in the qubit case perfect transposition can be achieved only for equatorial states since the map reduces to a $\pi$-rotation on the Bloch sphere.

\subsection{Process-optimized universal transposition cloning map}
\tab We can also define the ideal transposition cloning operation as the map that outputs
two copies of the transpose of the input state, i.e.
\begin{equation}
  \cE_{\ideal}:\rho \mapsto \rho^{T}\otimes \rho^T, \qquad\forall\rho\in \cD_{\pure}.
\end{equation}
Here we compute the process fidelity of the $1\to 2$ transposition
cloning map for an arbitraty pure state as its input.  Let $\mu$ be the induced Haar measure for the pure states $\cD_{\pure}$ and $\sigma$ be
the Haar measure on the set of unitary operators $\cU(\cH)$. 
Define the average Choi operator
\begin{equation}\label{Eq:transpose-avg-choi}
  \wt{J}(\cE) := \int_{\cU(\cH)} U^{\dagger\otimes 3} J(\cE)U^{\otimes 3}\;d\sigma(U).
\end{equation}
Notice that by the Haar invariance we get
\begin{align}
  &F_{\proc}(\cE_{\ideal}, \cE|\cD_{\pure})   \\
  &= \int_{\cD_{\pure}} \tr \[J(\cE)(\rho^{T})^{\otimes 3}\] \;d\mu(\rho)\notag\\
  &= \tr \[ \int_{\cU(\cH)} J(\cE)(U\op{0}{0}U^\dagger)^{\otimes 3} \;d\sigma(U) \]\\
  &= \tr \[ \int_{\cU(\cH)} U^{\dagger\otimes 3}J(\cE)U^{\otimes 3}\;d\sigma(U) \dop{0}^{\otimes 3}\]\\
  &= \tr \[\wt{J}(\cE)\dop{0}^{\otimes 3}\].
\end{align}
To find the optimal process fidelity, we exploit the fact that $\wt{J}(\cE)$
is $U^{\otimes 3}$-invariant.  We then obtain the following proposition as a direct consequence
of the results from work by Eggeling and Werner \cite{Eggeling2001}.

\begin{proposition}
  The optimal process fidelity of the $1\to 2$ universal transpose cloning map is $6/(d^2+3d+2)$.
\end{proposition}
\begin{proof}
  Eggeling and Werner's have shown that any quantum state $\rho$ satisfying $U^{\otimes 3}\rho U^{\dagger \otimes 3} = \rho$ can be uniquely
  expressed as
  \begin{equation} \label{Eq:transpose-sum}
    \wt{J}(\cE) = \sum_{i\in \{+,-,0,1,2,3\}} c_iR_i,\q c_i\in \R,
  \end{equation}
  where
  \begin{align}
    &R_+ = \frac{1}{6}(\1 + V_{(12)} + V_{(23)} + V_{(31)} + V_{(123)} + V_{(132)})\\
    &R_- = \frac{1}{6}(\1 - V_{(12)} - V_{(23)} - V_{(31)} + V_{(123)} + V_{(132)})\\
    &R_0 = \frac{1}{3}(2\cdot \1 - V_{(123)} - V_{(321)}),\\
    &R_1 = \frac{1}{3}(2\cdot \1 - V_{(31)} - V_{(12)}),\notag\\
    &R_2 = \frac{1}{\sqrt{3}}(V_{(12)}-V_{(31)}),\; R_3 = \frac{i}{\sqrt{3}}(V_{(123)} - V_{(321)}).
  \end{align}
  Furthermore, it satisfies the conditions (i) $c_+,c_-,c_0 \geq 0$, (ii)$c_1^2+c_2^2 + c_3^2 \leq c_0^2$, and (iii) $\tr[\rho R_+] + \tr[\rho R_-] + \tr[\rho R_-] = 1$.  Let $\wt{J}(\cE)$ be the average Choi operator defined as \eqref{Eq:transpose-avg-choi}.
  Then $\wt{J}(\cE)$ is invariant under $U^{\otimes 3}$ action, and 
  the preceding statement holds except replacing (iii) with
  \begin{equation}
   \tr_{23}[R_+\wt{J}(\cE)+R_-\wt{J}(\cE)+R_0\wt{J}(\cE)] = \1.
  \end{equation}
  Suppose $\wt{J}(\cE) = \sum c_kR_k$.
  Then the objective function to maximize becomes
  \begin{align}
    \tr\[\wt{J}(\cE)\dop{0}^{\otimes 3}\] = c_+,
  \end{align}
  while the constraints $\wt{J}(\cE) \geq 0$ is equivalent to $c_+,c_-,c_0 \geq 0$
  and $c_1^2+c_2^2+c_3^2 \leq c_0^2$ and the trace condition $\tr \wt{J}(\cE) = d$
  is equivalent to
  \begin{equation}
    \frac{d^2+3d+2}{6}c_+ + \frac{d^2-3d+2}{6}c_- + \frac{2(d^2-1)}{3}c_0 = 1.
  \end{equation}
  Solving this linear programming gives
\begin{equation}
    c_-=c_0=c_1=c_2=c_3 = 0,\qquad    c_+ = \frac{6}{d^2+3d+2}.
\end{equation}
\end{proof}
  Eggeling and Werner use Schur-Weyl duality to characterize
  $U^{\otimes 3}$-invariant positive matrices, and details
  of the proof are in \cite{Eggeling2003} and \cite{Eggeling2001}.

\subsection{Process-optimized phase-covariant transposition cloning map}
\tab We next compare the previous result with the same type of cloner except with the domain further restricted to $\mc{D}_\phase$.  The ideal map for phase covariant
transposition cloner is given by
\begin{equation}
\label{Eq:phase-covariant-transpose-cloner}
  \cE_{\ideal}:\dop{\btheta} \mapsto \dop{-\btheta}^{\otimes 2}\qquad\forall\dop{\btheta}\in \cD_{\phase}.
\end{equation}
Different from the cloning case, the average Choi operator is defined as
\begin{equation}
  \label{choi-tilde-def}
  \wt{J}(\cE) := \frac{1}{\abs{S_d}\abs{S_3}}\sum_{\pi\in S_d}\sum_{\sigma\in S_3}
    V_\sigma U_\pi^{\otimes 3} \cT(J(\cE)) U_\pi^{\dagger \otimes 3} V_\sigma^{\dagger},
\end{equation}
where the twirling operation is given by
\begin{equation}
  \mc{T}(X):=\int U(-\btheta)^{\otimes 3}\left(X\right)U(\btheta)^{\otimes 3} \;d\mu(\btheta).
\end{equation}
The positivity and trace condition of $\wt{J}$ remains the same and its invariance under the
permutations of the basis vector also remains the same. The main difference is that
$\wt{J}$ is invariant under $U(\btheta)^{\otimes 3}$ instead of $U(-\btheta)\otimes U(\btheta)^{\otimes 2}$,
and it is also invariant under permuting any subsystems instead of just 2 and 3.
Despite the difference, how we solve the optimal process fidelity is
analogous to that of the cloning case. Here we just state the result as a theorem
and include the proof in the appendix for interested readers.
\begin{theorem}\label{Thm:transpose-cloning} Let $\cE_{\ideal}$ be given by Eq. \eqref{Eq:phase-covariant-transpose-cloner}.
  Then, the optimal process fidelity of $\cE_{\ideal}$ is
  \begin{equation}
    F_{\proc}^*(\cE_{\ideal}|\mc{D}_\phase) = \begin{cases} 3/4 & d = 2\\ 6/d^2 & d \geq 3. \end{cases}
  \end{equation}
\end{theorem}

\subsection{Phase-covariant hybrid transposition cloning}
Now consider a hybrid transposition cloning map 
\begin{equation}
\label{Eq:hybrid-transpose}
\mc{E}_\ideal:\dop{\btheta} \mapsto \dop{\btheta}\otimes \dop{-\btheta}\qquad\forall\dop{\btheta}\in\mc{D}_\btheta.
\end{equation}
It is easy to see that the symmetries of the process fidelity
between the optimal map and the above ideal map are identical to that of the
cloning case, except we permute the system 1 and 3.
In other words, we can reuse the constraints (i)-(iii) in Lemma \ref{Lem:avg-choi} to
construct the average Choi map for the hybrid cloning
and make a modification for (iv) to be $\tr_{12}(\wt{J}(\cE_{\ideal})) = \1$,
instead of taking the partial trace over system 2 and 3. In fact,
$U(-\btheta)\otimes U(\btheta)^{\otimes 2}$ invariance guarantees the partial trace
over any two system of $\wt{J}(\cE)$ to be a scalar multiple of $\1$. 
Hence, we can conclude that the process fidelity of the hybrid map is $(2d-1)/d^2$.  Furthermore, 
if we let $\mc{E}_{\text{hybrid}}$ denote an optimal approximation of the hybrid transposition cloner (phase-covariant) and $\mc{E}_{\text{cloner}}$ an optimal approximation of the $1\to 2$ cloner (phase-covariant), then 
\begin{equation}
\wt{J}(\cE_{\text{hybrid}}) = V_{(13)}\wt{J}(\cE_{\text{clone}})V_{(13)}.
\end{equation}

It is interesting that the two maps
\begin{align}
\dop{\btheta} &\mapsto \dop{\btheta}\otimes \dop{-\btheta}\notag\\
\dop{\btheta} &\mapsto \dop{\btheta}\otimes \dop{\btheta}
\end{align} 
can be approximated with the same process fidelity.  One might not expect this since the second can be obtained from the first by applying a transpose on the first system, and as computed above, the transposition itself has a process fidelity of $2/d$.  Hence a composition of maps yields a highly non-optimal approximation of the $1\to 2$ phase-covariant cloner.  We make this comparison more explicitly in the next section.


\onecolumngrid
\vspace{2em}

\begin{figure}[h!]
  \centering
  \begin{subfigure}[b]{0.45\textwidth}
  \includegraphics[width=\textwidth]{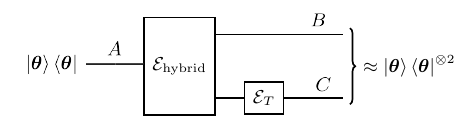}
  \caption{Modular cloning channel}
  \end{subfigure}
  \begin{subfigure}[b]{0.45\textwidth}
  \includegraphics[width=\textwidth]{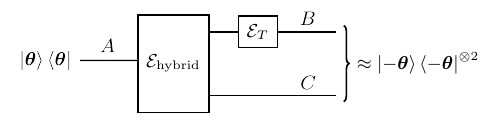}
  \caption{Modular transpose cloning channel}
  \end{subfigure}
  \caption{Modular cloning and transpose cloning channels}\label{fig:modular-cloner}
\end{figure}
\vspace{2em}

\twocolumngrid

\subsection{Modular building of QCMs}
\label{Sect:modular}

In this section we consider a modular approach of simulating some non-physical process by breaking the latter down into smaller parts and then simulating each of those.  We have just observed that a $1\to 2$ phase-covariant cloner can be obtained by combining the hybrid cloner (Eq. \eqref{Eq:hybrid-transpose}) with the transpose map (Eq. \eqref{Eq:phase-covariant-transpose}).
The construction is depicted in Fig.
\ref{fig:modular-cloner}, and the question we consider here is how well the combined
optimal maps for the individual parts, $\mc{E}_{\text{hybrid}}$ and $\mc{E}_{T}$ 
respectively, compare to the optimal map computed in Theorem \ref{Thm:1-2-cloner}.  
The composition of $\mc{E}_{\text{hybrid}}:A\to B\wt{C}$ and $\mc{E}_{T_C}:\wt{C}\to C$ can be expressed 
in terms of their Choi matrices as
\begin{multline}
 \hspace{-1em} (\cE_{T_C}\circ\cE_{\text{hybrid}})(\rho) = \tr_{A\wt{C}}\left[ \left(J(\cE_{\text{hybrid}})((\rho^T)^A \otimes \1^{B\wt{C}})\right)^{T_{\wt{C}}}\right.\\
\left. \otimes \1^{C} (J(\cE_T)^{C\wt{C}}\otimes \1^{AB}) \right],
\end{multline}

where $A,B,C$ denote the quantum systems indicated in Figure ~\ref{fig:modular-cloner}
and the subscript $T_{\wt{C}}$ denotes the partial transpose of system $\wt{C}$.

Similarly, the transposition cloning of Eq. \eqref{Eq:phase-covariant-transpose-cloner}
can be obtained by combining the hybrid cloner $\cE_{\text{hybrid}}:A\to \wt{B}C$ (Eq. \eqref{Eq:hybrid-transpose})
with the transpose map $\cE_{T_B}: \wt{B}\to B$ (Eq. \eqref{Eq:phase-covariant-transpose}). 
The construction is depicted in Fig.
\ref{fig:modular-anti-cloner}, and we compare to the optimal map computed in
Theorem \ref{Thm:transpose-cloning}. 
The composition of maps has the form
\begin{multline}
  \hspace{-1em} (\cE_{T_B}\circ\cE_{\text{hybrid}})(\rho)
  = \tr_{A\wt{B}}\left[ \left(J(\cE_{\text{hybrid}})
  ((\rho^T)^A \otimes \1^{\wt{B}C})\right)^{T_{\wt{B}}} \right. \\
  \left. \otimes \1^{B} (J(\cE_T)^{B\wt{B}}\otimes \1^{AC}) \right],
\end{multline}
where $A,B,C$ denote the quantum systems indicated in Figure ~\ref{fig:modular-cloner}
and the subscript $T_{\wt{B}}$ denotes the partial transpose of system $\wt{B}$. 

The process fidelity of both maps are
$\frac{3d-4}{d(d-1)(2d-1)}$, which is factor of $d$ smaller than the optimal cloner in Theorem \ref{Thm:1-2-cloner}.
However, it is comparable to that of optimal transposition cloning channel with process fidelity
of $6/d^2$. Hence in this case, the modular approach is asymptotically optimal (i.e. as $d\to\infty$) in the process fidelity.


\section{Conclusion}
\label{Sect:conclusion}

We have shown that the optimal process fidelity of the $1\to 2$
phase covariant-cloner is $\frac{2d-1}{d^2}$. The obtained channel yields its
single-qudit fidelity of $(d+1)/(2d-1)$, which is less than the previously known result
by \cite{FanEtAl2003}. Conversely, the cloner by \cite{FanEtAl2003} also does not yield the optimal process fidelity. This means that the optimal phase-covariant cloners do not coincide with two different
fidelity measures, as it is the case for the universal cloner.
We also defined a transpose cloning map, and showed that the optimal process fidelity
for transpose cloning of arbitrary pure states is $6/(d^2+3d+2)$ using the results from
\cite{Eggeling2001}.  In comparison, the optimal fidelity for transpose cloning of phase-covariant states is $6/d^2$.

Future work in this direction would include finding an explicit formula for $1\to M$ 
phase-covariant cloners.  This will require a representation-theoretic approach 
involving representations of $S_M$ with additional positivity and phase covariant 
structure.  In general, the $1\to M$ phase covariant cloner is interesting for 
applications in phase estimation.  We expect that the modular approach presented in 
Fig. \ref{fig:modular-cloner} may also be helpful in simulating such a map and other 
multi-system maps like it.



\appendix
\section{Solving the process fidelity for $1\to 2$ phase covariant transpose cloner}

Let $\cE_{\ideal}: \dop{-\btheta}\mapsto \dop{\btheta}^{\otimes 2}$ be the $1\to 2$
ideal phase covariant transpose cloning map.
\begin{lemma}\label{Lem:avg-choi-anti}
  Let $\cE$ be a quantum channel.
  Define the average Choi operator $\wt{J}$ for the channel $\cE$ as
  \begin{equation}
    \label{avg-choi-equatorial-2}
    \wt{J}(\cE) := \frac{1}{\abs{S_d}\abs{S_3}}\sum_{\substack{\pi\in S_d\\\sigma\in S_3}
    }
      V_\sigma U_\pi^{\otimes 3} \cT(J(\cE)) U_\pi^{\dagger \otimes 3} V_\sigma^{\dagger},
  \end{equation}
  where the twirling operation $\cT$ is defined as
\begin{equation}
  \label{phase-anti-proc}
  \mc{T}(X):=\int U(\btheta)^{\otimes 3}\left(X\right)U(-\btheta)^{\otimes 3} \;d\mu(\btheta).
\end{equation}

  then we have
  \begin{enumerate}[(i)]
    \item $\wt{J}(\cE)$ is invariant under conjugation by $U(\btheta)^{\otimes 3}$
      for any $\btheta$;\label{Eq:lem-Utheta-anti}
    \item $\wt{J}(\cE)$ is invariant under conjugation by $U_\pi^{\otimes 3}$
    for all $\pi \in S_d$;\label{Eq:lem-Upi-anti}
    \item $\wt{J}(\cE)$ is invariant under permuting subsystems;\label{Eq:lem-perm-anti}
    \item $\wt{J}(\cE)$ is positive and $\tr_{2}(\wt{J}(\cE)) = \1$.\label{Eq:lem-pos-anti}
  \end{enumerate}
\end{lemma}

  \begin{proof}
    (i) directly follows from 
      \begin{equation}
        U(\btheta)^{\otimes 3}\cT(J(\cE)) U(-\btheta)^{\otimes 3} = \cT(J(\cE)).
      \end{equation}
    Same as the cloning case, (ii), (iii) also directly follow from the definition of $\wt{J}$
    and (iv) follows from the fact that twirling operation, permutation of basis or
    the system do not change the positivity of $J(\cE)$.
    Trace is also invariant under $\cT$ and under conjugation by $U_\pi^{\otimes 3}$
    or $V_\sigma$.
  \end{proof}
With the above characterization of $\wt{J}$, we can establish the following in
an analogous manner to the cloning case.
\begin{lemma} \label{Lem:proc-anti}
  Suppose $\wt{J}$ is defined as \eqref{avg-choi-equatorial} for a quantum channel $\cE$.
  Then we have
  \begin{equation}
    F_{\proc}( \cE_{\ideal}, \cE|\cD_{\phase}) = \tr\left[\varphi_d^{+\otimes 3}\wt{J}(\mc{E})\right].
  \end{equation}
  We omit the proof since it is more or less identical to Lemma \ref{Lem:proc}
  for the cloning case.
\end{lemma}

Similar to the cloning case, we characterize a $d^3\times d^3$ hermitian operator
$X\neq 0$ that satisfies \ref{Eq:lem-Utheta-anti}, \ref{Eq:lem-Upi-anti},
\ref{Eq:lem-perm-anti}, \ref{Eq:lem-pos-anti} in Lemma \ref{Lem:avg-choi-anti}.
As a result, we can express $X$ as $\sum_{i=1}^6 x_iX_i$, where $x_i\in \C$ and
\begin{footnotesize}
\begin{align}
  &X_1 = \sum_{i} \op{iii}{iii}, \; X_2 = \sum_{i\neq k} \op{iik}{iik} + \op{iki}{iki} + \op{kii}{kii},\\
  &X_3 = \sum_{i\neq k} \op{iik}{kii} + \op{kii}{iik} + \op{kii}{iki} + \op{iik}{iki} + \op{iki}{iik}, \\
  &X_4 = \sum_{i\neq k\neq \ell} \op{ik\ell}{ik\ell}, \;
  X_5 = \sum_{i\neq k\neq \ell} \op{ik\ell}{k\ell i} + \op{ik\ell}{\ell ik},\\
  &X_6 = \sum_{i\neq k\neq \ell} \op{ik\ell}{ki\ell} + \op{ik\ell}{\ell ki} + \op{ik\ell}{i\ell k}.
\end{align}
\end{footnotesize}
If we denote $\mathbf{x} = (x_i) \in \C^6$, then
\begin{multline}
  \tr\left[\phi_d^{+\otimes 3} X\right]= \frac{1}{d^2}x_1 + \frac{(d-1)}{d^2} (3x_2+6x_3)\\
  + \frac{(d-1)(d-2)}{d^2}(x_4 + 2x_5 + 3x_6),
\end{multline}
and the trace condition
\small
\begin{equation}
 A(\mbf{x}) = x_1 + 3(d-1)x_2 + (d-1)(d-2)x_4 = 1.
\end{equation}

We construct the SDP the identical way as \eqref{SDP} to prove the optimality
of the process fidelity. Now we present the proof of the Theorem \ref{Thm:transpose-cloning}.
\begin{proof}
  When the dimension $d = 2$, $\bx = (0, 1/3, 1/3, 0, 0, 0)$ yields a primal feasible solution
  of $3/4$, and $\hat{Z} \oplus [3/4]$ is dual feasible, where
  $$ \hat{Z}= \frac{1}{4}(X_1+X_2+X_4) - \frac{1}{8}(X_3 + X_5 + X_6). $$
  Let $d\geq 3$ and define \begin{gather}
    k_d = \frac{1}{(d-1)(d-2)},\;
    \bx = (0, 0, 0, k_d, k_d, k_d).
  \end{gather}
  It can be easily be checked that the resulting operator satisfies the positivity
  and trace condition. Analogous to the cloning case, we prove the optimality of
  $F_\cE$ by finding a dual feasible $Z = \hat{Z}\oplus [z]$, where $z = 6/d^2$
  and $\hat{Z} = \sum_i b_iX_i$, $b_i\in \R$. The constraints $\tr F_iZ = c_i$ yields
  \begin{align}
      &\tr[X_1\hat{Z}] - z = - \frac{1}{d^2},\\
      &\tr[X_2\hat{Z}] - 3(d-1)z = -\frac{3(d-1)}{d^2},\\
      &\tr[X_3\hat{Z}] = -\frac{6(d-1)}{d^2},\\
      &\tr[X_4\hat{Z}] - (d-1)(d-2)z = -\frac{(d-1)(d-2)}{d^2}\notag\\
      &\tr[X_5\hat{Z}] = -\frac{2(d-1)(d-2)}{d^2},\\
      &\tr[X_6\hat{Z}] = -\frac{3(d-1)(d-2)}{d^2},
  \end{align}
  and we get
  \begin{equation}
    b_1 = b_2 = b_4 = \frac{5}{d^3},\q\q\q b_3 = b_5 = b_6 = -\frac{1}{d^3}.
  \end{equation}
  Next, we rewrite $\hat{Z}$ as the positive linear combination of projections to show positivity.
  Define the projections
  \begin{align}
      &P_A = \frac{2}{3}X_2 - \frac{1}{3}X_3\\
      &P_B = \frac{1}{2}\Big(1 + \frac{1}{\sqrt{2}}\Big)X_4 - \frac{1}{4\sqrt{2}}X_5 - \frac{1}{4\sqrt{2}}X_6.
  \end{align}
  Then we can write $\hat{Z}$ as
  \begin{align}
      \hat{Z} &=  \frac{5}{d^3}\Bigg(X_1 + X_2 + X_4\Bigg) - \frac{1}{d^3}\Bigg(X_3 + X_5 + X_6\Bigg)\notag\\
      &= \frac{1}{d^3}\left( 3P_1 + 4\sqrt{2}P_B + 5X_1 + 3X_2 + (3-2\sqrt{2})X_4 \right),
  \end{align}
  and hence $\hat{Z}\geq 0$ and also $Z\geq 0$. Therefore, $\frac{6}{d^2}$
  is indeed the optimal solution.

\end{proof}

\vfill

\bibliographystyle{abbrv}
\bibliography{process-fidelity}

\end{document}

%% file: packages.tex
\usepackage{fullpage} 
\usepackage{parskip}

\usepackage[english]{babel}
\usepackage[utf8]{inputenc}
\usepackage[pdftex, pdftitle={Article}, pdfauthor={Author}]{hyperref} 
\usepackage{array}
\usepackage{makecell}

\usepackage{amsmath,amsfonts,amsthm,amssymb}
\usepackage{hyperref}

\usepackage{tikz}
\usetikzlibrary{arrows}
\usetikzlibrary{quantikz}

\newcommand{\tab}{\hspace{2em}}

\usepackage{mathtools}
\mathtoolsset{centercolon}
\mathtoolsset{showonlyrefs,showmanualtags}

\usepackage{dsfont}

\definecolor{darkred}  {rgb}{0.5,0,0}
\definecolor{darkblue} {rgb}{0,0,0.5}
\definecolor{darkgreen}{rgb}{0,0.5,0}
\hypersetup{
  colorlinks = true,
  urlcolor  = blue,         
  linkcolor = darkblue,     
  citecolor = darkgreen,    
  filecolor = darkred       
}

\usepackage[shortlabels]{enumitem}

\newcommand{\ti}[1]{\textit{#1}}

\let\q\quad

\let\oldexists\exists
\renewcommand{\exists}{\oldexists\,}

\usepackage{thm-restate}
\theoremstyle{definition}
\newtheorem{theorem}{Theorem}
\newtheorem{proposition}{Proposition}
\newtheorem{lemma}{Lemma}

\newtheorem{definition}{Definition}

\let\oldmid\mid
\renewcommand{\mid}{\,\oldmid\,}

\newcommand{\mc}[1]{\mathcal{#1}}
\newcommand{\mbf}[1]{\mathbf{#1}}
\newcommand{\mbb}[1]{\mathbb{#1}}
\newcommand{\id}{\textrm{id}}
\newcommand{\wt}[1]{\widetilde{#1}}

\renewcommand{\(}{\left(}
\renewcommand{\)}{\right)}
\renewcommand{\[}{\left[}
\renewcommand{\]}{\right]}

\DeclarePairedDelimiter{\abs}{\lvert}{\rvert}%
\DeclarePairedDelimiter{\norm}{\lVert}{\rVert}%

\makeatletter
\let\oldabs\abs
\def\abs{\@ifstar{\oldabs}{\oldabs*}}
\let\oldnorm\norm
\def\norm{\@ifstar{\oldnorm}{\oldnorm*}}
\makeatother

\newcommand{\C}{\mathbb{C}}

\newcommand{\R}{\mathbb{R}}

\newcommand{\1}{\mathds{1}}

\newcommand{\cD}{\mathcal{D}}
\newcommand{\cE}{\mathcal{E}}

\newcommand{\cH}{\mathcal{H}}

\newcommand{\cK}{\mathcal{K}}

\newcommand{\cT}{\mathcal{T}}
\newcommand{\cU}{\mathcal{U}}

\renewcommand{\phi}{\varphi}

\renewcommand{\phi}{\varphi}

\DeclareMathOperator{\tr}{Tr}

\newcommand{\op}[2]{|#1\rangle \langle #2|} 
\newcommand{\dop}[1]{\ket{#1}\bra{#1}}

\renewcommand{\(}{\left(}
\renewcommand{\)}{\right)}
\renewcommand{\[}{\left[}
\renewcommand{\]}{\right]}